\documentclass[journal,onecolumn,12pt]{IEEEtran}
\usepackage{pslatex}
\usepackage{amsfonts,color,morefloats}
\usepackage{amssymb,amsmath,latexsym,amsthm}
\usepackage{lineno}

\newcommand{\wt}{{\mathrm{wt}}}

\newcommand{\comp}{{\mathrm{comp}}}

\newcommand{\Z}{\mathbb{{Z}}}

\newcommand{\tr}{{\mathrm{Tr}}}
\newcommand{\gf}{{\mathrm{GF}}}
\newcommand{\f}{{\mathbb{F}}}

\newcommand{\C}{{\mathcal{C}}}

\newcommand{\bc}{{\mathbf{c}}}

\newtheorem{theorem}{Theorem}

\newtheorem{corollary}[theorem]{Corollary}

\newtheorem{example}{Example}

\begin{document}

\title{A Construction of Linear Codes over $\f_{2^t}$ from Boolean Functions \thanks{The research of K. Feng was supported by NSFC No. 11471178, 11571007 and the Tsinghua National Lab. for Information Science and Technology.}}

\author{Can Xiang
\thanks{C. Xiang is with the College of Mathematics and Information Science, Guangzhou University, Guangzhou 510006, Guangdong Province, China. Email: cxiangcxiang@hotmail.com.} and
Keqin Feng\thanks{K. Feng is with the Department of Mathematical Sciences, Tsinghua University, Beijing, 100084, China.
Email: kfeng@math.tsinghua.edu.cn.} and
Chunming Tang\thanks{C. Tang is with the College of Mathematics and Information Science, Guangzhou University, Guangzhou 510006, Guangdong Province, China. Email: ctang@gzhu.edu.cn.}
}

\date{\today}
\maketitle

\begin{abstract}
In this paper, we present a construction of linear codes over $\f_{2^t}$ from Boolean functions, which is a generalization of Ding's method \cite[Theorem 9]{Ding15}. Based on this construction, we give two classes of linear codes $\tilde{\C}_{f}$ and $\C_f$ (see Theorem \ref{thm-maincode1} and Theorem \ref{thm-maincodenew}) over $\f_{2^t}$ from a Boolean function $f:\f_{q}\rightarrow \f_2$, where $q=2^n$ and $\f_{2^t}$ is some subfield of $\f_{q}$. The complete weight enumerator of $\tilde{\C}_{f}$ can be easily determined from the Walsh spectrum of $f$, while the weight distribution of the code $\C_f$ can also be easily settled. Particularly, the number of nonzero weights of $\tilde{\C}_{f}$ and $\C_f$ is the same as the number of distinct Walsh values of
$f$. As applications of this construction, we show several series of linear codes over $\f_{2^t}$ with two or three weights by using bent, semibent, monomial and quadratic Boolean function $f$.

\end{abstract}


\section{Introduction}\label{sec-intro}

Throughout this paper, let $q=2^n$ for any positive integer $n\geq2$, $\f_q$  be the finite field with
$q$ elements and $\f_q^*=\f_q  \backslash \{0\}$. By viewing $\f_q$ as a $n$-dimensional vector space $\f_2^{~n}$ with respect to a fixed $\f_2$-basis of $\f_q$, a function $f$ from $\f_{2^n}$  (or $\f_2^{~n}$) to $\f_2$ is called a Boolean function with $n$ variables, which can be expressed by
$$
f=f(x):\f_q \rightarrow \f_2.
$$
The Walsh transform of $f$ is defined by
\begin{eqnarray}\label{eqn-w}
W_f(y)=\sum_{x \in \f_q} (-1)^{f(x)+\tr(xy)} \in \Z
\end{eqnarray}
where $y\in \f_{2^n}$ and $\tr$ is the trace function from $\f_{2^n}$ to $\f_2$. The Walsh spectrum of $f$ is the following multiset
$$
\{W_f(y):y\in \f_{2^n}\}.
$$
For convenience, we denote the Walsh spectrum of $f$ by
\begin{eqnarray}\label{eqn-walshspectrum}
[w_1]^{m_1}[w_2]^{m_2}\cdots[w_r]^{m_r}
\end{eqnarray}
if the multiset has $r$ distinct values $w_1,w_2,\cdots,w_r$, where $m_i=|\{y \in \f_q: W_f(y)=w_i\}|$ for each positive integer $i$ satisfying $1\leq i\leq r$. It is obvious that $m_1+m_2+\cdots+m_r=q$.

Recently, C.Ding \cite{Ding15} presented a construction of binary linear code $\C_f$ from a Boolean function $f$ such that the weight distribution of $\C_f$ can be derived directly from the Walsh spectrum of $f$ \cite[Theorem 9]{Ding15}. The number of distinct nonzero weights of $\C_f$ is the same as the number $r$ of distinct Walsh values of $f$. Particularly, from bent and semibent functions he gets a series of binary linear codes with two and three weights, respectively. In this paper, we show a generalization (see Theorem \ref{thm-maincode1}) of Ding's construction \cite{Ding15} by the following way. Firstly, from any Boolean function $f:\f_q\rightarrow \f_2$ we construct two classes of linear codes
$\tilde{\C}_{f}$ and $\C_{f}$ over some subfield $\f_{2^t}$ of $\f_q$ for some $t$ satisfying a certain condition. Secondly, the complete weight
enumerator of $\tilde{\C}_{f}$ and the weight enumerator of $\C_{f}$ can be determined in terms of the Walsh spectrum of $f$. Particularly, we show that the code $\tilde{\C}_{f}$ is `nearly' a composition constant code which means that for each codeword $c=(c_1,c_2,...,c_{N}) \in \tilde{\C}_{f}$, $N_c(\alpha)$ are the same for all $\alpha \in \f_{2^t}^*$, where
\begin{eqnarray}\label{eqn-NC}
N_c(\alpha)=|\{1\leq i \leq N: c_i=\alpha\}|
\end{eqnarray}
is the number of $\alpha$-components of the codeword $c$.

Now we recall some terminology on linear codes. An $[N,\, k]$ linear code $\C$ over $\f_{2^t}$ is a subspace of $\f_{2^t}^{~N}$ with dimension ${\mathrm{dim}}_{\f_{2^t}}\C=k$. For each codeword $c=(c_1,c_2,...,c_{N}) \in \C$ and $\alpha \in \f_{2^t}$, let $N_c(\alpha)$ be the number of $\alpha$-components of the codeword $c$. Namely, $N_c(\alpha)$ is defined by (\ref{eqn-NC}). Then $\sum_{\alpha \in \f_{2^t}}N_c(\alpha)=N$. The complete weight enumerator of $\C$ is a polynomial on $x_\alpha(\alpha \in \f_{2^t})$ defined by
\begin{eqnarray}\label{eqn-complete}
\textrm{K}_{\C}(X)=\textrm{K}_{\C}(x_\alpha:~\alpha \in \f_{2^t})=\sum_{c\in \C}~\prod_{\alpha \in \f_{2^t}} x_{\alpha} ^ {~N_{c}(\alpha)} \in \Z[x_\alpha: \alpha \in \f_{2^t}].
\end{eqnarray}
This is a homogeneous polynomial with degree $N$. Let $x_0=1$ and $x_\alpha=x$ for all $\alpha \in \f_{2^t}^{~*}$ in (\ref{eqn-complete}), we get the weight enumerator of $\C$:
$$
\textrm{E}_{\C}(x)=1+\sum_{i=1}^{N}A_ix^i \in \Z[x],
$$
where
$$A_i=|\{c\in \C:\textrm{wt}_{\textrm{H}}(c)=\sum_{\alpha \in \f_{2^t}^*}N_c(\alpha)=i\}|$$
is the number of codewords $c$ in $\C$ with the Hamming weight $\textrm{wt}_{\textrm{H}}(c)=i$. $(1,A_1,\ldots,A_N)$  is called the weight distribution of $\C$.
The code $\C$ is said to be $v$ weights if the number of nonzero $A_i$ in the sequence $(A_1, A_2, \cdots, A_N)$ is equal to $v$.
For binary case $t=1$, it is obvious that the complete weight enumerator of $\C$ is essentially the same as the weight enumerator. When $t\geq 2$,
the complete weight enumerator $\textrm{K}_{\C}(X)$ of $\C$ gives more information on the code $\C$ than the weight enumerator of $\C$.
%An $[N,\, k,\,d]$ linear code $\C$ is called {\em optimal} if its parameters $N$, $k$ and $d$ meet
%a bound on linear codes \cite[Chapter 2]{HP}.  An $[N,\, k,\,d]$ linear code $\C$ is called {\em almost optimal} if $[n,\,  \kappa,\, d+1]$ meets a bound on linear codes \cite[Chapter 2]{HP}.

This paper is organized as follows. In Section \ref{sec-mainresult1}, we present a generalization (see Theorem \ref{thm-maincode1}) of Ding's construction. From a Boolean function $f:\f_q\rightarrow \f_2$ we construct two classes of linear codes
$\tilde{\C}_{f}$ and $\C_{f}$ over $\f_{2^t}$, while the complete weight
enumerator of $\tilde{\C}_{f}$ and the weight enumerator of $\C_{f}$ can be determined in terms of the Walsh spectrum of $f$, where $t$ is a divisor of $n$ satisfying a certain condition. If $f$ has $r$ distinct Walsh values, then both of $\tilde{\C}_{f}$ and $\C_{f}$ have $r$ distinct
nonzero weights. As applications of Theorem \ref{thm-maincode1}, Section \ref{sec-application} gives several series of linear codes $\C_{f}$ over $\f_{2^t}$ with two or three weights constructed by using bent, semibent, monomial and quadratic Boolean function $f:\f_q\rightarrow \f_2$. Moreover, we also present a method (see Theorem \ref{thm-maincodenew}) by concatenation based on Theorem \ref{thm-maincode1}, and get more linear codes over $\f_{2^t}$ with three weights and more flexible length. The last Section \ref{sec-concluding} is conclusion.

\section{A Construction of Linear Codes over $\f_{2^t}$}\label{sec-mainresult1}

The objective of this section is to present a generalization of Ding's construction for linear codes and construct two classes of linear codes $\tilde{\C}_{f}$ and $\C_f$ over $\f_{2^t}$ from a Boolean function $f$, where $t$ is a divisor of $n$ satisfying a certain condition.

Let $f=f(x)$ be a Boolean function from $\f_{2^n}$ to $\f_2$. In this section, we consider a subfield $\f_{2^t}$ of $\f_{2^n}$, where $t$ satisfies
the following condition:
\begin{eqnarray}\label{dfn-condition}
\mbox{ $f(\alpha x)=f(x) $ for all $\alpha \in \f_{2^t}^*$ and $x\in \f_q$ }.
\end{eqnarray}
Namely, $f$ is a constant on each coset of subgroup $\f_{2^t}^*$ in $\f_q^*$.

Define
\begin{eqnarray}\label{eqn-maincodedefset}
\tilde{D}=\tilde{D}(f)=\{x \in \f_{q}^*: f(x)=0\}.
\end{eqnarray}
The condition (\ref{dfn-condition}) implies that the set $\tilde{D}$ is an union of $N$ distinct cosets $x_1\f_{2^t}^{~*},\cdots, x_N\f_{2^t}^{~*}$ in $\f_q^*$. Therefore, $\tilde{N}=|\tilde{D}|=(2^t-1)N$. Let $D=\{x_1,x_2,\cdots, x_N\}$ and $\tilde{D}=\{x_1,x_2,\cdots, x_N,x_{N+1},\cdots,x_{\tilde{N}}\}$. Now we define two classes of linear codes over $\f_{2^t}$ as follows:
\begin{eqnarray}\label{eqn-maincode1f}
\tilde{\C}=\tilde{\C}_{f}=\{\tilde{c}_b=(\tr_t^n(bx_1), \tr_t^n(bx_2), \ldots, \tr_t^n(bx_{\tilde{N}})): b \in \f_{q}\}
\end{eqnarray}
and
\begin{eqnarray}\label{eqn-maincode2f}
\C=\C_{f}=\{c_b=(\tr_t^n(bx_1), \tr_t^n(bx_2), \ldots, \tr_t^n(bx_{N})): b \in \f_{q}\},
\end{eqnarray}
where $\tr_t^n$ is the trace function from $\f_{2^n}$ to $\f_{2^t}$.

The following theorem is the main result of this paper.

\begin{theorem}\label{thm-maincode1}
Let $q=2^n$ and $f(x):\f_q \rightarrow \f_2$ be a Boolean function with the Walsh spectrum $[w_1]^{m_1}[w_2]^{m_2}\cdots[w_r]^{m_r}$. Assume that
\begin{itemize}
  \item $W_f(0)=w_i$ for some $i$ with $1\leq i \leq r$;
  \item $w_i-w_j\neq -q$ for all positive integer $j$, $1\leq j \leq r$; and
  \item $t$ is an positive integer satisfying $t|n$ and the condition (\ref{dfn-condition}).
\end{itemize}
Then
\begin{enumerate}
  \item the linear code $\tilde{\C}=\tilde{\C}_{f}$ of (\ref{eqn-maincode1f}) has parameters $[\tilde{N}, n/t]$ and the complete weight enumerator
\begin{eqnarray}\label{eqn-completetheorem}
\textrm{K}_{\tilde{\C}}(X)=&~x_0^{\tilde{N}}+(m_{i}-1)x_0^{N_i}\prod_{\gamma \in \f_{2^t}^*}x_\gamma^{~2^{-(t+1)}q}+ %~\\ &
\sum_{1\leq j\leq r(j\neq i)}m_{j}~x_0^{N_j}(\prod_{\gamma \in \f_{2^t}^*}x_\gamma)^{M_j},
\end{eqnarray}
where
\begin{eqnarray*}
\tilde{N}&=&2^{n-1}+\frac{1}{2}(w_i-1-(-1)^{f(0)}) \nonumber \\
N_j&=&2^{-t-1}(q+w_i+(2^t-1)w_j)-2^{-1}(1+(-1)^{f(0)})\nonumber \\
M_j&=&2^{-t-1}(q+w_i-w_j)\nonumber
\end{eqnarray*}
for each positive integer $j$ within $1\leq j\leq r$ and $i$ is determined by $W_f(0)=w_i$. The weight enumerator of $\tilde{\C}$ is
\begin{eqnarray}\label{eqn-1weighttheorem}
\textrm{E}_{\tilde{\C}}(x)=1+(m_i-1)x^{(2^t-1)q/2^{t+1}}+\sum_{j=1,j\neq i}^{r}m_j~x^{(2^t-1)M_j}.
\end{eqnarray}
\item the linear code $\C=\C_{f}$ of (\ref{eqn-maincode2f}) has parameters $[N, n/t]$ and the weight enumerator of $\C$ is
\begin{eqnarray}\label{eqn-2weighttheorem}
\textrm{E}_{\C}(x)=1+(m_i-1)x^{q/2^{t+1}}+\sum_{j=1,j\neq i}^{r}m_j~x^{M_j},
\end{eqnarray}
where $N=\tilde{N}/(2^t-1)$.
\end{enumerate}
Particularly, both of $\tilde{\C}$ and $\C$ have $r$ distinct (nonzero) weights.
\end{theorem}
\begin{proof} \emph{1)}~By definition and assumption , the length of $\tilde{\C}$ is
\begin{eqnarray}\label{eqn-N}
\tilde{N}
&=& \left|\left\{x \in \f_q^{~*}: f(x)=0\right\}\right| \nonumber \\
&=& \frac{1}{2} \sum_{y \in \f_2} \sum_{x \in \f_q} (-1)^{yf(x))}-\frac{1}{2}(1+(-1)^{f(0)}) \nonumber \\
&=& \frac{1}{2} \sum_{x \in \f_q} (1+(-1)^{f(x)})-\frac{1}{2}(1+(-1)^{f(0)}) \nonumber \\
&=& \frac{1}{2}(q+W_f(0)-1-(-1)^{f(0)}) \nonumber \\
&=& 2^{n-1} +  \frac{1}{2}(w_i-1-(-1)^{f(0)}) \nonumber.
\end{eqnarray}

Next we calculate the number
$$
N_{\tilde{c}(b)}(\gamma)=|\{x \in \f_q^*: f(x)=0 \mbox{ and } \tr_{t}^{n}(bx)=\gamma\}|
$$
for any $b \in \f_q^*$ and any $\gamma \in \f_{2^t}$.

By definition and the basic facts of additive characters, we have
\begin{eqnarray}\label{eqn-wab1}
N_{\tilde{c}(b)}(\gamma) &=& \frac{1}{2^{t+1}} \sum_{x \in \f_q^*}(1+(-1)^{f(x)}) \sum_{\alpha \in \f_{2^t}} (-1)^{\tr_1^t(\alpha(\tr_t^n(bx)+\gamma))} \nonumber \\
&=& \frac{1}{2^{t+1}} \sum_{\alpha \in \f_{2^t}} (-1)^{\tr_1^t(\alpha \gamma)} \sum_{x \in \f_q^*}(1+(-1)^{f(x)})(-1)^{\tr(bx\alpha)}  \nonumber \\
&=& \frac{1}{2^{t+1}}(\sum_{\alpha \in \f_{2^t}} (-1)^{\tr_1^t(\alpha \gamma)} \sum_{x \in \f_q}(1+(-1)^{f(x)})(-1)^{\tr(bx\alpha)}- \nonumber \\
& & \sum_{\alpha \in \f_{2^t}} (-1)^{\tr_1^t(\alpha \gamma)}(1+(-1)^{f(0)}))  \nonumber  \\
&=& \left\{ \begin{array}{ll}
        \frac{1}{2^{t+1}}(\sum_{\alpha \in \f_{2^t}} (-1)^{\tr_1^t(\alpha \gamma)} \sum_{x \in \f_q}(1+(-1)^{f(x)})(-1)^{\tr(bx\alpha)}) & \mbox{ if $\gamma \neq 0$} \nonumber \\
        \frac{1}{2^{t+1}}(\sum_{\alpha \in \f_{2^t}} \sum_{x \in \f_q}(1+(-1)^{f(x)})(-1)^{\tr(bx\alpha)}-
2^t(1+(-1)^{f(0)}))  & \mbox{ if $\gamma=0$}
\end{array}
\right. \nonumber \\
&=& \left\{ \begin{array}{ll}
        \frac{1}{2^{t+1}}(\sum_{x \in \f_q}(1+(-1)^{f(x)})+\sum_{\alpha \in \f_{2^t}^*}(-1)^{\tr_1^t(\alpha \gamma)}W_f(b \alpha)) & \mbox{ if $\gamma \neq 0$} \\
        \frac{1}{2^{t+1}}(\sum_{x \in \f_q}(1+(-1)^{f(x)})+\sum_{\alpha \in \f_{2^t}^*}W_f(b\alpha))-
\frac{1}{2}(1+(-1)^{f(0)})  & \mbox{ if $\gamma=0$}
\end{array}
\right.
\end{eqnarray}

But for any $b \in \f_q^*$ and any $\alpha \in \f_{2^t}^*$, we have
\begin{eqnarray}\label{eqn-wab}
W_f(\alpha b)&=& \sum_{x \in \f_q} (-1)^{f(x)+\tr(\alpha bx)} \nonumber \\
&=& \sum_{x' \in \f_q} (-1)^{f(\alpha^{-1}x')+\tr(bx')} ~~~~~~~~~(\mbox{ Let $x'=\alpha x$})\nonumber \\
&=& \sum_{x' \in \f_q} (-1)^{f(x')+\tr(bx')} ~~~~~~~~~~~~~(\mbox{ By the condition (\ref{dfn-condition}}))\nonumber \\
&=& W_f(b).
\end{eqnarray}

In view of (\ref{eqn-wab}), formula (\ref{eqn-wab1}) becomes that
\begin{eqnarray}\label{eqn-simple}
N_{\tilde{c}(b)}(\gamma)
&=& \left\{ \begin{array}{ll}
        \frac{1}{2^{t+1}}(q+W_f(0)-W_f(b)) & \mbox{ if $\gamma \neq 0$} \nonumber \\
        \frac{1}{2^{t+1}}(q+W_f(0)+(2^{t}-1) W_f(b))-
\frac{1}{2}(1+(-1)^{f(0)})  & \mbox{ if $\gamma=0$}
\end{array}
\right. \nonumber \\
&=& \left\{ \begin{array}{ll}
        \frac{1}{2^{t+1}}(q+w_i-w_j) & \mbox{ if $\gamma \neq 0$}  \\
        \frac{1}{2^{t+1}}(q+w_i+(2^{t}-1) w_j)-
\frac{1}{2}(1+(-1)^{f(0)})  & \mbox{ if $\gamma=0$}
\end{array}
\right.
\end{eqnarray}
By assumption that $W_f(0)=w_i$, $W_f(b)=w_j$ and the walsh spectrum of $f$ is $[w_1]^{m_1}[w_2]^{m_2}\cdots[w_r]^{m_r}$, we have
\begin{eqnarray}\label{eqn-codeweight2}
W_f(b)=\left\{ \begin{array}{ll}
w_1                               & \mbox{ occurs $m_1$ times,} \\
\vdots                               & \mbox{ $\vdots$} \\
w_i                               & \mbox{ occurs $m_i-1$ times,} \\
w_{i+1}                               & \mbox{ occurs $m_{i+1}$ times,} \\
\vdots                               & \mbox{ $\vdots$} \\
w_r                               & \mbox{ occurs $m_r$ times,}
\end{array}
\right.
\end{eqnarray}
when $b$ runs through $\f_q^*$. Thus we get the formula (\ref{eqn-completetheorem}) on $\textrm{K}_{\tilde{\C}}(X)$. The formula (\ref{eqn-1weighttheorem}) can be derived directly by taking $x_0=1$ and $x_\gamma=x$ for all $\gamma \in \f_{2^t}^*$ in $\textrm{K}_{\tilde{\C}}(X)$. Meanwhile, by assumption that $w_i-w_j\neq-q$ for all positive integer $j$ satisfying $1\leq j\leq r$, we know that the Hamming weight of $\tilde{c}(b)$ is
$$
\wt_{\textrm{H}}(\tilde{c}(b))=(2^t-1)M_j=\frac{2^t-1}{2^t}(q+w_i-w_j)\neq 0
$$
for each $b \in \f_q^*$.
This means that the code $\tilde{\C}_f$ has $q$ distinct codewords. Hence, the dimension of the code $\tilde{\C}_f$ is $\textrm{dim}_{\f_{2^t}}{\tilde{\C}}=\textrm{log}_{2^t}q=\frac{n}{t}$.

\emph{2)} For each $b\in \f_q$, we have the codeword $\tilde{c}(b)=(\tr_t^{n}(bx))_{x\in \tilde{D}}$ in the code $\tilde{\C}$ and
$c(b)=(\tr_t^{n}(bx))_{x\in D}$ in the code $\C$, where $\tilde{D}=\bigcup_{x\in D}~x\f_{2^t}^*$ is a disjoint union of cosets $x\f_{2^t}^*~(x\in D)$. It is obvious that $\tilde{N}=|\tilde{D}|=(2^t-1)|D|=(2^t-1)N$. For each element $x\alpha~(x\in D,\alpha \in \f_{2^t}^*)$, $\tr_t^{n}(bx\alpha)=\alpha \tr_t^{n}(bx)$. Therefore, $\tr_t^{n}(bx\alpha)\neq 0$ if and only if $\tr_t^{n}(bx)\neq 0$. This implies that
$\wt_{\textrm{H}}(\tilde{c}(b))=(2^t-1)\wt_{\textrm{H}}(c(b))$ for all $b\in \f_q^*$. Then the formula (\ref{eqn-2weighttheorem}) on weight enumerator $\textrm{E}_{\C}(x)$ can be derived from the formula (\ref{eqn-1weighttheorem}). This completes the proof of
Theorem \ref{thm-maincode1}.
\end{proof}

Let $t=1$. It is clear that the condition (\ref{dfn-condition}) is true. Form Theorem \ref{thm-maincode1} we get the following corollary on binary linear codes which essentially is \cite[Theorem 9]{Ding15}.
\begin{corollary}\label{thm-maincode1t1}
Let $f=f(x): \f_q \rightarrow \f_2$ be a Boolean function with Walsh spectrum $[w_1]^{m_1}[w_2]^{m_2} \cdots [w_r]^{m_r}$. Let $\tr$ be the trace mapping from $\f_q$ to $\f_2$. Assume that $W_f(0)=w_i$ and $w_i-w_j\neq-q$ for all $j~(1\leq j \leq r)$. Let $D=D_f=\{x_1,x_2,\cdots,x_N\}$ be the set of zeros of $f(x)$ in $\f_q^*$, $N=|D|$, Then the subset $\C=\C_f$ of $\f_2^{N}$ defined by
$$
\C=\C_f=\{c(b)=(\tr(bx_1),\tr(bx_2),\cdots, \tr(bx_N)):b\in \f_q\}
$$
is a binary linear code with the parameters  $[N, n]$ and
the weight enumerator
$$
1+(m_i-1)~x^{q/4}+\sum_{1\leq j\leq r,~j\neq i}~m_j~x^{(q+w_i-w_j)/4},
$$
where $N=2^{n-1}+\frac{1}{2}(w_i-1-(-1)^{f(0)})$.
\end{corollary}

We remark that  Theorem \ref{thm-maincode1} open a way to construct linear codes $\tilde{\C}_f$ and $\C_f$ over $\f_{2^t}$ with $r$ distinct nonzero weights from any Boolean function $f$ with $r$ distinct Walsh values.

\section{Linear codes over $\f_{2^t}$ with two or three weights}\label{sec-application}

In this section, we show several examples on linear codes with two or three weights as applications of Theorem \ref{thm-maincode1}. Moreover, we also give a method by concatenation based on Theorem \ref{thm-maincode1}, and get more linear codes with three weights.

For t=1, Ding has presented such examples on binary linear codes from bent, semibent, almost bent and quadratic Boolean functions \cite[Corollary 10, 11, 13 and Theorem 14]{Ding15}. Thus we focus on construction of linear codes over $\f_{2^t}$ for the case $t\geq 2$.
Since the parameters of the code $\C_{f}$ is better of the code $\tilde{\C}_{f}$ in Theorem \ref{thm-maincode1}, from now on we focus on the linear code
$\C_{f}$.

\subsection{Two weights linear codes over $\f_{2^t}$}
In this subsection, we construct several classes of linear codes over $\f_{2^t}$ with two weights by using bent functions.

Let $f$ be a Boolean function from $\f_q$ to $\f_2$ throughout this subsection. If $f$ is bent and $n=2m$, then the Walsh spectrum of $f$ is
$$
[2^m]^{2^{n-1}+2^{m-1}(-1)^{f(0)}}[-2^m]^{2^{n-1}-2^{m-1}(-1)^{f(0)}}.
$$

The conclusion of the following theorem is straightforward from Theorem \ref{thm-maincode1} and its proof is omitted.

\begin{theorem}\label{thm-two1}
Let $f$ be a bent function from $\f_q$ to $\f_2$, $q=2^n$, $n=2m$ and $W_f(0)=\varepsilon \cdot 2^m$ with $\varepsilon\in \{1,-1\}$. If $\f_{2^t}$ is a subfield of $\f_q$ where $t|n$ and the condition (\ref{dfn-condition})
is satisfied, then the code $\C=\C_{f}$ defined by (\ref{eqn-maincode2f}) is a two weights linear code over $\f_{2^t}$ with the
parameters $[\frac{1}{2^t-1}(2^{n-1}+\frac{1}{2}(\varepsilon 2^m -1-(-1)^{f(0)})),n/t]$ and the weight enumerator
      $$
      \textrm{E}_{\C}(x)=1+(2^{n-1}+\varepsilon\cdot 2^{m-1}(-1)^{f(0)}-1)~x^{2^{n-t-1}} +(2^{n-1}-\varepsilon\cdot 2^{m-1}(-1)^{f(0)})~x^{2^{n-t-1}+\varepsilon\cdot 2^{m-t}}.
      $$
\end{theorem}

It is obvious that Theorem \ref{thm-two1} essentially is \cite[Corollary 10]{Ding15} if $t=1$.

Many bent functions have been found since Rothaus published his original paper in 1976 \cite{Rothaus 1976}. We refer
the reader to \cite{BCHKM2012, CM2011, Mesnager2014, Mesnager2015, Mesnager2011A, LHTK2013, LK2006, YG2006A} and the references therein for detailed information.
Here we only consider two cases which fit with Theorem \ref{thm-two1} for the case $t\geq 2$.

\emph{Case 1: Monomial Polynomials.}

The following corollary is a direct consequence of Theorem \ref{thm-two1}.

\begin{corollary}\label{thm-two1cor}
Let $f(x)=\tr(\alpha x^d)$ be a bent function from $\f_q$ to $\f_2$, where $q=2^n$, $n=2m$, $2\leq d\leq q-2$ and $\alpha \in \f_q^*$.
Suppose that $t|n$ and $(2^t-1)|d$. Then
the code $\C_{f}$ defined by (\ref{eqn-maincode2f}) is a two weights linear code over $\f_{2^t}$ with the
  parameters $[\frac{1}{2^t-1}(2^{n-1}+\frac{1}{2}(\varepsilon 2^m -2)),n/t]$ and the weight enumerator
      $$
      \textrm{E}_{\C}(x)=1+(2^{n-1}+\varepsilon\cdot 2^{m-1}-1)~x^{2^{n-t-1}} +(2^{n-1}-\varepsilon\cdot 2^{m-1})~x^{2^{n-t-1}+\varepsilon\cdot 2^{m-t}},
      $$
\end{corollary}
where $\varepsilon \cdot 2^m=W_f(0)=\sum_{x\in \f_q}(-1)^{\tr(\alpha x^d)}$.

Let $\theta$ be a generator of $\f_q^*$. The known monomial bent functions $f(x)=\tr(\alpha x^d)$ can be summarized in Table \ref{tab-listbent} for $\alpha \in \f_q^* $ and even $n=2m$.
\begin{table}[ht]
\begin{center}
\caption{The monomial bent functions $f(x)=\tr(\alpha x^d)$}\label{tab-listbent}
\begin{tabular}{|c|c|c|c|} \hline
Name of index $d$ & $d$ &  Conditions & Reference  \\ \hline
Gold &$2^{h}+1$ & $2|\frac{n}{gcd(n,h)}$, $\alpha \neq 0$ and $\alpha \notin <\theta^{2^h+1}>$ &  \cite{Leander2006} \\ \hline
Dillon &$2^{m}-1$ & $\alpha \in \f_{2^m}^*$ and $\sum_{x\in \f_{2^m}^*}(-1)^{\tr_1^m(x^{-1}+\alpha x)}=-1$ &  \cite{Leander2006} \\ \hline
Kasami &$2^{2h}-2^h+1$    & $3\nmid m$, $gcd(h,n)=1$, $\alpha \neq 0$ and $\alpha \notin <\theta^{3}>$ &  \cite{Leander2006} \\ \hline
Leander &$(2^{h}+1)^2$ &  $n=4h$ and $2\nmid h$&  \cite{Leander2006} \\ \hline
CCK &$2^{2h}+2^h+1$  &  $n=6h$, $\alpha \in \gf(2^{3h})^*$ and $\tr_h^{3h}(\alpha)=0$   & \cite{CCK2008} \\ \hline
\end{tabular}
\end{center}
\end{table}

In order to construct linear codes over $\f_{2^t}$ from the monomial bent functions by using Corollary \ref{thm-two1cor},
we need $t$ satisfying the following extra condition (except the conditions listed in Table \ref{tab-listbent})
\begin{eqnarray}\label{dfn-conditionextra}
\mbox{$t|n$ and $(2^t-1)|d$}
\end{eqnarray}
since $(2^t-1)|d$ implies the condition (\ref{dfn-condition}).
\begin{example}
In order to construct linear codes over $\f_4$ by using bent functions of Table \ref{tab-listbent}, we need the extra condition which is reduced to be $3|d$ for $t=2$. Namely, the extra condition is $2\nmid h$ for Gold and Kasami bent functions, $2|h$ for CCK bent functions and $2|m$ for Dillon bent functions in the Table \ref{tab-listbent}. No extra condition is needed for Leander bent functions. Put this extra condition in Table \ref{tab-listbent} we get two-weight linear codes
over $\f_4$ with the parameter and weight enumerator given by Corollary \ref{thm-two1cor}.
\end{example}

\begin{example}\label{exp-bent}
For each divisor $t\geq2$ of $m$, we have $(2^t-1)|(2^m-1)$. Therefore from Dillon bent function $f(x)$ listed in Table \ref{tab-listbent} we can construct a two-weight linear code over $\f_{2^t}$ with the parameter and weight enumerator given by Corollary \ref{thm-two1cor} for any divisor $t$ of $m$.
\end{example}

\emph{Case 2: Quadratic Functions.}

Let the quadratic Boolean function $f:\f_q \rightarrow \f_2$ be defined by
$$
f(x)=\sum_{i=1}^{m-1}\tr(c_ix^{1+2^i})+\tr_1^m(c_mx^{1+2^m}), \ \ \ \ \ \mbox{$c_m\in \f_{2^m}$, $c_i\in\f_q$, $1\leq i\leq m-1$},
$$
where $n=2m$ and $q=2^n$. Many quadratic bent functions have been obtained by using quadratic form theory. We refer
the reader to \cite{YG2006A,CPT2005} and the references therein for detailed information. Here we consider some cases fitting in
Theorem \ref{thm-two1} for the case $t=2$.

Let $n=2m$. For $t=2$, if $f(x)$ is a quadratic bent function with the following form

\begin{eqnarray}\label{dfn-qform1}
f(x)=\sum_{\lambda=0}^{[\frac{m}{2}]-1}\tr(c_\lambda x^{1+2^{2\lambda+1}}), \ \  \mbox{$c_\lambda \in\f_q$, $0\leq \lambda \leq [\frac{m}{2}]-1$}
\end{eqnarray}
or
\begin{eqnarray}\label{dfn-qform2}
f(x)=\sum_{\lambda=0}^{s-1}\tr(c_\lambda x^{1+2^{2\lambda+1}})+\tr_1^m(c_mx^{1+2^m}), \ \  \ \mbox{$m=2s+1$,
$c_\lambda \in\f_q$, $0\leq \lambda \leq s-1$,
$c_m\in \f_{2^m}$},
\end{eqnarray}
then all index $1+2^{2\lambda+1}$ and
$1+2^m$ for odd $m$ can be divided by $3=2^2-1$. From such quadratic bent functions $f(x)$ in (\ref{dfn-qform1}) and (\ref{dfn-qform2}), we can
construct linear codes over $\f_4$ by Theorem \ref{thm-two1}.

\begin{example} As an special case of (\ref{dfn-qform2}), we consider the following form:
\begin{eqnarray}\label{dfn-qform2special}
f(x)=\sum_{\lambda=0}^{s-1}c_\lambda \tr(x^{1+2^{2\lambda+1}})+\tr_1^m(x^{1+2^m}), \ \  \ \mbox{$n=2m$, $m=2s+1$,
$c_\lambda \in\f_2$, $0\leq \lambda \leq s-1$}.
\end{eqnarray}
Furthermore, the simple special case of (\ref{dfn-qform2special}) is
\begin{eqnarray}\label{dfn-qform2special2}
f(x)=\tr(x^{1+2^{2\lambda+1}})+\tr_1^m(x^{1+2^m}), \ \  \ \mbox{$n=2m$, $m=2s+1$,
$0\leq \lambda \leq s-1$}.
\end{eqnarray}

Let $c(x)=\sum_{\lambda=0}^{s-1}c_\lambda(x^{2\lambda+1}+x^{m-(2\lambda+1)})+x^m\in\f_2[x]$. By \cite[Corollary 1]{YG2006A},
we have
\begin{itemize}
  \item the function $f(x)$ of (\ref{dfn-qform2special}) is a bent function from $\f_q$ to $\f_2$ if and only if $gcd(c(x),x^m+1)=1$;
  \item the function $f(x)$ of (\ref{dfn-qform2special2}) is a bent function from $\f_q$ to $\f_2$ if and only if $gcd(3(2\lambda+1),m)=1$.
\end{itemize}

When $f(x)$ is a bent function with the form (\ref{dfn-qform2special}) or (\ref{dfn-qform2special2}), then the linear code constructed by Theorem \ref{thm-two1} is a two weights linear code over $\f_4$.
\end{example}

\subsection{Three weights linear codes over $\f_{2^t}$}

For most of known Boolean functions $f(x): \f_q \rightarrow \f_2$ with three Walsh values, the values of $W_f(w)$ are $w_1=0$, $w_2=A$
and $w_{3}=-A$ and $A=2^l$ for some $l$ satisfying $l\geq [\frac{n}{2}]+1$. If $l=[\frac{n}{2}]+1=m+1$ for $n=2m$ or $n=2m+1$, the function $f$ is called semibent.

Let
the Walsh spectrum of $f$ be $[w_1]^{m_1}[w_2]^{m_2}[w_{3}]^{m_{3}}$, where $(w_1,w_2,w_{3})=(0,A,-A)$ and
$$
m_i=|\{w\in \f_q: W_f(w)=w_i\}|
$$
for each $i\in\{1,2,3\}$. Then we have the following system of
equations on $m_i(i=1,2,3)$:
\begin{eqnarray}\label{eqn-wtdsemibentfcode6c}
\left\{
\begin{array}{lll}
m_1+m_2+m_3 &=& q, \\
m_1w_1+m_2w_2+m_3w_3 &=&\sum_{x\in\f_q,w\in\f_q}(-1)^{f(x)+\tr(wx)}= q\cdot(-1)^{f(0)}, \\
m_1w_1^{~2}+m_2w_2^{~2}+m_3w_3^{~2} &=&\sum_{w\in\f_q}(\sum_{x\in\f_q}(-1)^{f(x)+\tr(wx)})^2= q^2,
\end{array}
\right.
\end{eqnarray}
Solving this system of equations gets

\begin{eqnarray}\label{eqn-wtdsemibentfcode6c1}
\left\{
\begin{array}{lll}
m_1&=& q-(\frac{q}{A})^2 \\
m_2&=& \frac{1}{2}((\frac{q}{A})^2+\frac{q}{A}(-1)^{f(0)})\\
m_3&=& \frac{1}{2}((\frac{q}{A})^2-\frac{q}{A}(-1)^{f(0)}).
\end{array}
\right.
\end{eqnarray}

Thus, the Walsh spectrum  of $f(x)$ is
\begin{eqnarray}\label{dfn-threewalsh}
[0]^{q-(\frac{q}{A})^2}[A]^{\frac{1}{2}((\frac{q}{A})^2+\frac{q}{A}(-1)^{f(0)})}[-A]^{\frac{1}{2}((\frac{q}{A})^2-\frac{q}{A}(-1)^{f(0)})}.
\end{eqnarray}

By Theorem \ref{thm-maincode1}, we get the following three-weight linear codes $\C_f$  over $\f_{2^t}$.
\begin{theorem}\label{thm-three}
Suppose that  $q=2^n$ and $f(x)$ is a Boolean function from $\f_q$ to $\f_2$ with the Walsh spectrum defined by (\ref{dfn-threewalsh}).
Let $(w_1,w_2,w_{3})=(0,A,-A)$, $W_f(0)=w_i$ for some $i\in\{1,2,3\}$ and $(m_1,m_2,m_3)$ be given by (\ref{eqn-wtdsemibentfcode6c1}). Assume that $w_i-w_j\neq-q$ for each $j\in\{1,2,3\}$ and $j\neq i$.
If $t|n$ and the condition (\ref{dfn-condition})
is satisfied, then
the code $\C_{f}$ of (\ref{eqn-maincode2f}) is a three weights linear code over $\f_{2^t}$ with the
  parameters $[\frac{1}{2^t-1}(2^{n-1}+\frac{1}{2}(w_i -1-(-1)^{f(0)})),n/t]$ and the weight enumerator
      $$
      E_{\C_f}(x)=1+(m_i-1)~x^{2^{n-t-1}} +\sum_{1\leq j\leq 3,~j\neq i}m_j~x^{2^{n-t-1}+\frac{1}{2}(w_i-w_j)}.
      $$
\end{theorem}

Many Boolean functions with three Walsh values have been found. Here we only consider two cases.

\emph{Case 1: Monomial Polynomials.}

For any nonzero integer $b$, let $v_2(b)\geq 0$ be the 2-adic exponential valuation of $b$ defined by $2^{v_2(b)}|b$ and $2^{v_2(b)+1}\nmid b$.
The known monomial Boolean functions $f(x)=\tr(x^d)$ with three Walsh values $\{0,A,-A\}$ can be summarized in Table \ref{tab-listxdthree} which are listed in \cite{AKL2015}, where $q=2^n$ and $1\leq d\leq q-1$. We refer
the reader to \cite{AKL2015} and the references therein for detailed information.
\begin{table}[ht]
\begin{center}
\caption{$f(x)=\tr(x^d)$ with three Walsh values $\{0,A,-A\}$}\label{tab-listxdthree}
\begin{tabular}{|c|c|c|c|} \hline
Series & $d$ &  Conditions & $A$  \\ \hline
\uppercase\expandafter{\romannumeral1} &$2^{h}+1$ & $v_2(h)\geq v_2(n)$ &  $\sqrt{2^{gcd(h,n)}q}$ \\ \hline
\uppercase\expandafter{\romannumeral2}  &$2^{2h}-2^h+1$    &$v_2(h)\geq v_2(n)$ &  $\sqrt{2^{gcd(h,n)}q}$ \\ \hline
\uppercase\expandafter{\romannumeral3} & $2^m+2^{\frac{m+1}{2}}+1$  & $n=2m,~2\nmid m$ &  $2\sqrt{q}$ \\ \hline

\uppercase\expandafter{\romannumeral4}  & $2^{\frac{m+1}{2}}+3$ & $n=2m, ~2\nmid m $ & $2\sqrt{q}$ \\ \hline
\uppercase\expandafter{\romannumeral5}  & $2^{m}+3$ &  $n=2m+1$ &  $\sqrt{2q}$ \\ \hline
\uppercase\expandafter{\romannumeral6} &$2\cdot 3^h+1$  & $ n| 4h+1$ & $\sqrt{2q}$ \\ \hline
\end{tabular}
\end{center}
\end{table}

In Table \ref{tab-listxdthree}, it is clear that the series (\uppercase\expandafter{\romannumeral3}-\uppercase\expandafter{\romannumeral6}) are
semibent functions, and the series (\uppercase\expandafter{\romannumeral1}-\uppercase\expandafter{\romannumeral2}) are semibent functions if and only if $gcd(h,n)=1$ or $gcd(h,n)=2$.

For $f(x)=\tr(x^d)$, we have $f(0)=0$. Thus, the length of the code $\C_{f}$ constructed in Theorem \ref{thm-three} becomes $\frac{1}{2(2^t-1)}(2^{n}+w_i-2)$. In order to construct three weights linear codes $\C_{f}$ over
$\f_{2^t}$ by using $f(x)=\tr(x^d)$ listed in Table \ref{tab-listxdthree}, we need extra conditions $t|n$ and $2^t-1|d$ which imply the condition (\ref{dfn-condition}).

\begin{example}\label{exp-t5}
In order to construct the code $\C_{f}$ over $\f_{2^5}$ by using
$f(x)=\tr(x^{2\cdot 3^h+1})$ of the series (\uppercase\expandafter{\romannumeral6}) in Table \ref{tab-listxdthree}, the extra conditions are
$5|n$ and $31|2\cdot 3^h+1$ plus the conditions $n|4h+1$ and $2\cdot 3^h+1\leq 2^n-1$. By a computation of elementary number theory, these
conditions are reduced to be $h=30l+21$, $5|n$, $n|120l+85$ and $2\cdot 3^h+1\leq 2^n-1$, where $l$ is any non-negative integer.
\end{example}

%\subsection{Linear codes over $\f_{2^t}$ with more flexible length}\label{sec-mainresult2}

To get linear codes over $\f_{2^t}$ with more flexible length, now we give a method as a consequence of Theorem \ref{thm-maincode1}. Based on this method, we can get more linear codes over $\f_{2^t}$ with three weights.

Let $q_1=2^{n_1}$ and $q_2=2^{n_2}$ for any positive integers $n_1\geq 2$ and $n_2\geq 2$.
Define a Boolean function
\begin{eqnarray}\label{eqn-definef}
f(x_1,x_2)=f_1(x_1)+f_2(x_2): \f_{q_1}\times \f_{q_2} \rightarrow \f_2,
\end{eqnarray}
where $f_1(x_1): \f_{q_1}\rightarrow \f_2$ and $f_2(x_2): \f_{q_2} \rightarrow \f_2$ are
Boolean functions.

By definition, we know that the Walsh values of $f$ are
$$
W_f(y_1,y_2)=\sum_{(x_1,x_2)\in \f_{q_1}\times \f_{q_2}}(-1)^{f(x_1,x_2)}(-1)^{\tr_1^{n_1}(x_1y_1)+\tr_1^{n_2}(x_2y_2)}
=W_{f_1}(y_1)W_{f_2}(y_2)
$$
for all $(y_1,y_2)\in \f_{q_1}\times \f_{q_2}$. From this we know that if the Walsh spectrums of $f_1$ and $f_2$ are $[w_1]^{m_1}[w_2]^{m_2}\cdots [w_{r_1}]^{m_{r_1}}$ and
$[\omega_1]^{\mu_1}[\omega_2]^{\mu_2}\cdots [\omega_{r_2}]^{\mu_{r_2}}$ respectively, then
the Walsh spectrum of $f$ is $\prod_{1\leq i\leq r_1, 1\leq j\leq r_2 }[w_i\omega_j]^{m_i\mu_j}$ which can be reduced to be
$$
[\Omega_1]^{M_1}[\Omega_2]^{M_2} \cdots [\Omega_{r}]^{M_{r}} ~~~~(\Omega_1 < \Omega_2 < \cdots < \Omega_{r}, M_i\geq 1, 1\leq i\leq r)
$$
by using $[w]^{m'}[w]^{m''}=[w]^{m'+m''}$. Thus the number of the Walsh values of $f$ may be less than $r_{1}r_{2}$.

We define
\begin{eqnarray}\label{eqn-defsetnew}
\tilde{D}=\tilde{D}_f=\{(x_1,x_2)\in \f_{q_1} \times \f_{q_2} \backslash \{(0,0)\}: f(x_1,x_2)=0\}
\end{eqnarray}
and $\tilde{N}=|\tilde{D}|$. Let $\f_{2^t}$ be a subfield of $\f_{q_1}$ and $\f_{q_2}$, namely, $t|gcd(n_1,n_2)$. Then $\f_{q_1}\times \f_{q_2}$ is
a $\f_{2^t}$-vector space by $\alpha (x_1,x_2)=(\alpha x_1,\alpha x_2)$ for all $\alpha \in \f_{2^t}$, $x_1 \in \f_{q_1}$ and $x_2 \in \f_{q_2}$. The dimension of the vector space $\f_{q_1}\times \f_{q_2}$ over $\f_{2^t}$ is $\frac{1}{t}(n_1+n_2)$. Moreover, if $t$ satisfies
the following condition:
\begin{eqnarray}\label{dfn-condition1}
\mbox{$f_1(\alpha x_1)=f_1(x_1)$ and $f_2(\alpha x_2)=f_2(x_2)$ for all $\alpha \in \f_{2^t}$, $x_1\in \f_{q_1}$ and $x_2\in \f_{q_2}$ },
\end{eqnarray}
then
\begin{eqnarray}\label{dfn-condition2}
\mbox{$f(\alpha (x_1,x_2))=f(\alpha x_1,\alpha x_2)=f_1(\alpha x_1)+f_2(\alpha x_2)=f_1(x_1)+f_2(x_2)=f(x_1,x_2)$ }.
\end{eqnarray}
Therefore
the set $\tilde{D}$ of (\ref{eqn-defsetnew}) is a disjoint union of $N$ sets $P^{(i)}\f_{2^t}^*$ if the condition (\ref{dfn-condition1}) is satisfied, where
$P^{(i)}=(x_1^{(i)},x_2^{(i)})\in \f_{q_1}\times \f_{q_2}$, $1\leq i \leq N$ and $N=\frac{\tilde{N}}{2^t-1}$.
We define
\begin{eqnarray}\label{eqn-defsetnew1}
D=D_f=\{P^{(1)},P^{(2)},\cdots, P^{(N)}\}
\end{eqnarray}
and the following linear code over $\f_{2^t}$ by
\begin{eqnarray}\label{eqn-maincodenew}
\C=\C_f=\{\bc_{b}: b=(b_1,b_2)\in \f_{q_1} \times \f_{q_2} \},
\end{eqnarray}
where
\begin{eqnarray}\label{eqn-codeword2}
\bc_{b}=(\tr_t^{n_1}(b_1x_1^{(1)})+\tr_t^{n_2}(b_2x_2^{(1)}), \tr_t^{n_1}(b_1x_1^{(2)})+\tr_t^{n_2}(b_2x_2^{(2)})
,\cdots,\tr_t^{n_1}(b_1x_1^{(N)})+\tr_t^{n_2}(b_2x_2^{(N)}))\in \f_{2^t}^N.
\end{eqnarray}
%Recall that $\gf(2^t)$ is a subfield of $\gf(q_1)$ and $\gf(q_2)$.

We have the following theorem. Its proof is similar to that
of Theorem \ref{thm-maincode1} and is omitted.

\begin{theorem}\label{thm-maincodenew}
Let $f$ be a Boolean function given by (\ref{eqn-definef}), $t|gcd(n_1,n_2)$, $t$ satisfies the condition (\ref{dfn-condition1}), and let the Walsh spectrums of $f_1$ and $f_2$ be $[w_1]^{m_1}[w_2]^{m_2}\cdots [w_{r_1}]^{m_{r_1}}$ and
$[\omega_1]^{\mu_1}[\omega_2]^{\mu_2}\cdots [\omega_{r_2}]^{\mu_{r_2}}$, respectively. Assume that the Walsh spectrums of $f(x_1,x_2)=f_1(x_1)+f_2(x_2)$ is reduced to be
$
[\Omega_1]^{M_1}[\Omega_2]^{M_2} \cdots [\Omega_{l}]^{M_{l}},
$
where $M_{\lambda}\geq 1$ for each $\lambda \in \{1,2,\cdots,l\}$ and all $\Omega_{\lambda}~(1\leq \lambda \leq l)$ are distinct
integers. Then
the code $\C_f$ of (\ref{eqn-maincodenew}) is a $[\frac{1}{2(2^t-1)}(q_1q_2+\Omega_i-1-(-1)^{f_1(0)+f_2(0)}),\frac{n_1+n_2}{t}]$ linear code over $\f_{2^t}$
with the weight enumerator
\begin{eqnarray}\label{eqn-weight}
1+(M_i-1)~x^{\frac{q_1q_2}{2^{t+1}}}+\sum_{1\leq j\leq l,~j\neq i}M_jx^{P_j},
\end{eqnarray}
where $P_j=\frac{1}{2^{t+1}}(q_1q_2+\Omega_i-\Omega_j)$ for each $j$ within $1\leq j \leq l$ and $i$ is determined by $\Omega_i=W_f(0,0)=W_{f_1}(0)W_{f_2}(0)$ for some $i$ within $1\leq i\leq l$.
\end{theorem}

%\begin{proof}
%As calculated in the proof of Lemma \ref{lem-length}, the length of $\tilde{\mathcal{C}}_{D}$ is
%\begin{eqnarray}
%|D|
%&=&|\{(x_1,x_2)\in \gf(q_1)\times \gf(q_2)/ \{0,0\}: f(x_1,x_2)=0\}| \nonumber \\
%&=&\frac{1}{2}(q_1q_2+W_f(0,0)-1-(-1)^{f_1(0)+f_2(0)}),\nonumber
%\end{eqnarray}
%where $W_f(0,0)=W_{f_1}(0)W_{f_2}(0)$. Meanwhile, the complete weight enumerator and weight enumerator of $\tilde{\mathcal{C}}_{D}$ can also be calculated as shown in the proof of Theorem \ref{thm-maincode1}. This completes the proof.
%\end{proof}

As special cases of Theorem \ref{thm-maincodenew}, we give some examples as follows.

\begin{example}\label{exp-parti}
Let $n_1$ be even, $f_1(x_1)$ be a bent function, and $f_2(x_2)=\tr_1^{n_2}(cx_2)$ with $c\in \f_{q_2}^*$. Then the Walsh spectrums of $f_1$ and $f_2$ are $[\sqrt{q_1}]^{\frac{1}{2}(q_1+\sqrt{q_1})}[\sqrt{-q_1}]^{\frac{1}{2}(q_1-\sqrt{q_1})}$ and $[0]^{q_2-1}[q_2]^1$, respectively. Thus the the Walsh spectrum of $f$ is
$[0]^{q_1(q_2-1)}[q_2\sqrt{q_1}]^{\frac{1}{2}(q_1+\sqrt{q_1})}[-q_2\sqrt{q_1}]^{\frac{1}{2}(q_1-\sqrt{q_1})}$. Let $t=1$. Then the linear code $\C_f$ of (\ref{eqn-maincodenew}) constructed by Theorem \ref{thm-maincodenew} is a binary linear code with the following three weights:
$$
\frac{1}{2}(q_1q_2+\Omega_i-\Omega_1),~~\frac{1}{2}(q_1q_2+\Omega_i-\Omega_2), ~~\frac{1}{2}(q_1q_2+\Omega_i-\Omega_3),
$$
and the weight enumerator can be given by (\ref{eqn-weight}),
where $(\Omega_1,\Omega_2,\Omega_3)=(0,q_2\sqrt{q_1},-q_2\sqrt{q_1})$ and $i$ is determined by $\Omega_i=W_{f_1}(0)W_{f_2}(0)$.
\end{example}

Note that this kind function $f$ in Example \ref{exp-parti} is called partially-bent function in \cite{Carlet1993}(see \cite[Theorem 2.1]{Carlet1993}).

\begin{example}
Let $t=5$, $n_1=2m$, $5|m$, $l$ be a non-negative integer, $h=30l+21$, $5|n_2$, $n_2|120l+85$ and $2\cdot 3^h+1\leq 2^{n}-1$. We have a bent function $f_1(x_1)=\tr_1^{n_1}(\alpha x_1^{2^m-1})$ in Example \ref{exp-bent} and a Boolean function
$f_2(x_2)=\tr_1^{n_2}(x_2^{~2\cdot 3^h+1})$ in Example \ref{exp-t5}. It is clear that $\f_{2^5}$ is a subfield of $\f_{q_1}$ and $\f_{q_2}$. The condition (\ref{dfn-condition1}) is satisfied since $2^5-1|2^m-1$ and $2^5-1|2\cdot 3^h+1$. By definition, Examples \ref{exp-bent} and \ref{exp-t5} and Table \ref{tab-listxdthree}, we know that the Walsh spectrums of $f_1$ and $f_2$ are
$[\sqrt{q_1}]^{m_+}[-\sqrt{q_1}]^{m_-}$ and $[0]^{\mu_0}[\sqrt{2q_2}]^{\mu_+}[-\sqrt{2q_2}]^{\mu_-}$, respectively. Thus the Walsh spectrum of
$f(x_1,x_2)=f_1(x_1)+f_2(x_2)$ is
$$
[0]^{(m_+ + m_-)\mu_0} [\sqrt{2q_1q_2}]^{m_+ \mu_+ + m_- \mu_-}[-\sqrt{2q_1q_2}]^{m_+ \mu_- + m_- \mu_+},
$$
where $(m_+,m_-)=(\frac{1}{2}(q_1+\sqrt{q_1}),\frac{1}{2}(q_1-\sqrt{q_1}))$ and $\{\mu_0,\mu_+,\mu_-\}$ is determined by (\ref{eqn-wtdsemibentfcode6c1}). Therefore the code $\C_f$ constructed by Theorem \ref{thm-maincodenew} is a three weights linear code over $\f_{2^5}$ and the weight enumerator of the code $\C_f$ is given by (\ref{eqn-weight}).
\end{example}

\begin{example}
Let $A_1\geq 1$ and $A_2\geq 1$, and let $f_1(x_1): \f_{q_1} \rightarrow \f_2$ and $f_2(x_2): \f_{q_2} \rightarrow \f_2$ be
Boolean functions with the Walsh spectrums of $[0]^{m_0}[A_1]^{m_+}[-A_1]^{m_-}$ and $[0]^{\mu_0}[A_2]^{\mu_+}[-A_2]^{\mu_-}$, respectively. Then the Walsh spectrum of $f(x_1,x_2)=f_1(x_1)+f_2(x_2)$ is
$
[0]^{M_0}[A_1A_2]^{M_+}[-A_1A_2]^{M_-},
$
where
$$
(M_0,M_+,M_-)=(m_0(\mu_0+ \mu_+ + \mu_-)+(m_+ + m_-)\mu_0,~m_+\mu_+ +m_-\mu_-,~m_+\mu_- +m_-\mu_+).
$$
If $\f_{2^t}$ is a subfield of $\f_{q_1}$ and $\f_{q_2}$, the code $\C_f$ constructed by Theorem \ref{thm-maincodenew} is a three weights linear code over $\f_{2^t}$ and the weight enumerator of the code $\C_f$ is given by (\ref{eqn-weight}).
\end{example}

\emph{Case 2: Quadratic Functions.}

For all quadratic functions with the following form:
$$
f(x)=\sum_{i=1}^{[\frac{n}{2}]}\tr(c_i~x^{1+2^i}), ~~~(c_i\in\f_q, q=2^n),
$$
the Walsh spectrum of $f(x)$ is $[0]^{2^n-2^{2h}}[2^{n-h}]^{2^{2h-1}+2^{h-1}}[-2^{n-h}]^{2^{2h-1}-2^{h-1}}$, where $h$ can be
determined by quadratic form theory (see \cite{CPT2005}). If $2h=n$, then $f$ is a bent function. Otherwise, $2h< n$, $f$ has three Walsh values and the code $\C_{f}$ is a binary linear code with three weights. The weight enumerators of
$\C_{f}$ is given by Theorem \ref{thm-three}. Remark that $f$ is semibent if and only if $n-h=[\frac{n+1}{2}]$.

Recently, some functions $f$ with three Walsh values $(w_1,w_2,w_3)$ and $w_i\neq 0$ for each $i\in\{1,2,3\}$ have been found
(see \cite[Theorem 3.2]{LY2015}). For such functions $f$, the binary linear code $\C_{f}$ has three weights and its weight enumerator
can be calculated by Theorem \ref{thm-maincode1}.

\section{Concluding remarks}\label{sec-concluding}

In this paper, we generalized the construction of linear codes of Ding's method. Base on this generalization, we constructed two classes of linear codes $\tilde{\C}_{f}$ and $\C_f$ (see Theorem \ref{thm-maincode1}) over $\f_{2^t}$ from a Boolean function $f:\f_q\rightarrow \f_2$  and their weight distributions determined by the Walsh spectrum of this Boolean function $f$, where $q=2^n$ and $\f_{2^t}$ is some subfield of $\f_{q}$.
Moreover, the parameters of more linear codes over $\f_{2^t}$ with a few weights can be derived directly.
Particularly, a number of classes of two-weight and
three-weight codes are derived from some known classes of bent, semibent, monomial and quadratic Boolean functions.
Instead of writing down all these codes, we documented a few classes of them as examples in this paper. In addition,
the complete weight enumerator of the linear code $\tilde{\C}_{f}$  and the weight enumerator of the linear code $\C_f$ presented
in this paper were settled in simple way in terms of the Walsh spectrum of $f$.

%The linear codes presented in this paper are interesting, as linear codes with a few weights have applications in
%secret sharing \cite{ADHK,CDY05,YD06,DingDing2015}, authentication codes \cite{CX05}, combinatorial designs and graph
%theory \cite{CG84,CK85}, and association schemes \cite{CG84}. If
%linear codes are employed to construct authentication codes having new parameters
%via the framework in \cite{DHKW,CX05}, we need to know not only
%the weight distribution of the linear codes, but also the complete weight enumerator of the linear codes.
%Another advantage of the linear codes in this paper is that their complete weight enumerator was
%settled. In the literature the complete weight enumerator of only a
%few classes of linear codes is known.

For further generalization and research in $\f_{p^n}$ with prime $p\geq 3$ case, we refer the reader to \cite{ZLFH2015}, \cite{TLQZT2015} and \cite{XC2015}.

%Compared with other linear codes with a few weights, the construction method of the codes in this paper
%is generalized. This makes the analysis of the linear codes much
%easier.

%\section*{Acknowledgements}
%The authors are very grateful to the reviewers and the Associate Editor, Dr. Yongyi Mao, for their comments and suggestions
%that improved the presentation and quality of this paper.

\end{document}